\documentclass[a4paper,cleveref,UKenglish]{lipics-v2019}

\usepackage{microtype}
\usepackage[utf8]{inputenc}
\usepackage{enumerate}
\usepackage[capitalise]{cleveref} 
\usepackage{mathtools,textcomp}
\usepackage{booktabs}
\usepackage{arydshln}
\usepackage[vlined]{algorithm2e}\DontPrintSemicolon\LinesNumbered
\crefname{algorithm}{Algorithm}{Algorithms}
\crefname{algocf}{Algorithm}{Algorithms}

\SetCommentSty{commentFont}

\usepackage{newunicodechar}
\usepackage{silence}
\newunicodechar{⊥}{\bot}
\newunicodechar{✓}{\checkmark}
\newunicodechar{⁻}{^-}
\newunicodechar{⁺}{^+}
\newunicodechar{₋}{_-}
\newunicodechar{₊}{_+}
\newunicodechar{ℓ}{\ell}
\newunicodechar{•}{\item}
\newunicodechar{…}{\dots}
\newunicodechar{≔}{\coloneqq}
\newunicodechar{≤}{\leq}
\newunicodechar{≥}{\geq}
\newunicodechar{≰}{\nleq}
\newunicodechar{≱}{\ngeq}
\newunicodechar{⊕}{\oplus}
\newunicodechar{≠}{\neq}
\newunicodechar{¬}{\neg}
\newunicodechar{≡}{\equiv}
\newunicodechar{₀}{_0}
\newunicodechar{₁}{_1}
\newunicodechar{₂}{_2}
\newunicodechar{₃}{_3}
\newunicodechar{₄}{_4}
\newunicodechar{₅}{_5}
\newunicodechar{₆}{_6}
\newunicodechar{₇}{_7}
\newunicodechar{₈}{_8}
\newunicodechar{₉}{_9}
\newunicodechar{⁰}{^0}
\newunicodechar{¹}{^1}
\newunicodechar{²}{^2}
\newunicodechar{³}{^3}
\newunicodechar{⁴}{^4}
\newunicodechar{⁵}{^5}
\newunicodechar{⁶}{^6}
\newunicodechar{⁷}{^7}
\newunicodechar{⁸}{^8}
\newunicodechar{⁹}{^9}
\newunicodechar{∈}{\in}
\newunicodechar{∉}{\notin}
\newunicodechar{⊂}{\subset}
\newunicodechar{⊃}{\supset}
\newunicodechar{⊆}{\subseteq}
\newunicodechar{⊇}{\supseteq}
\newunicodechar{⊄}{\nsubset}
\newunicodechar{⊅}{\nsupset}
\newunicodechar{⊈}{\nsubseteq}
\newunicodechar{⊉}{\nsupseteq}
\newunicodechar{∪}{\cup}
\newunicodechar{∩}{\cap}
\newunicodechar{∀}{\forall}
\newunicodechar{∃}{\exists}
\newunicodechar{∄}{\nexists}
\newunicodechar{∨}{\vee}
\newunicodechar{∧}{\wedge}
\newunicodechar{ℝ}{\mathbb{R}}
\newunicodechar{ℕ}{\mathbb{N}}
\newunicodechar{𝔽}{\mathbb{F}}
\newunicodechar{ℤ}{\mathbb{Z}}
\newunicodechar{·}{\cdot}
\newunicodechar{∘}{\circ}
\newunicodechar{×}{\times}
\newunicodechar{↑}{\uparrow}
\newunicodechar{↓}{\downarrow}
\newunicodechar{→}{\rightarrow}
\newunicodechar{←}{\leftarrow}
\newunicodechar{⇒}{\Rightarrow}
\newunicodechar{⇐}{\Leftarrow}
\newunicodechar{↔}{\leftrightarrow}
\newunicodechar{⇔}{\Leftrightarrow}
\newunicodechar{↦}{\mapsto}
\newunicodechar{∅}{\emptyset}
\newunicodechar{∞}{\infty}
\newunicodechar{≅}{\cong}
\newunicodechar{≈}{\approx}

\newunicodechar{α}{\alpha}
\newunicodechar{β}{\beta}
\newunicodechar{γ}{\gamma}
\newunicodechar{Γ}{\Gamma}
\newunicodechar{δ}{\delta}
\newunicodechar{Δ}{\Delta}
\newunicodechar{ε}{\varepsilon}
\newunicodechar{ζ}{\zeta}
\newunicodechar{η}{\eta}
\newunicodechar{θ}{\theta}
\newunicodechar{Θ}{\Theta}
\newunicodechar{ι}{\iota}
\newunicodechar{κ}{\kappa}
\newunicodechar{λ}{\lambda}
\newunicodechar{Λ}{\Lambda}
\newunicodechar{μ}{\mu}
\newunicodechar{ν}{\nu}
\newunicodechar{ξ}{\xi}
\newunicodechar{Ξ}{\Xi}
\newunicodechar{π}{\pi}
\newunicodechar{Π}{\Pi}
\newunicodechar{ρ}{\rho}
\newunicodechar{σ}{\sigma}
\newunicodechar{Σ}{\Sigma}
\newunicodechar{τ}{\tau}
\newunicodechar{υ}{\upsilon}
\newunicodechar{ϒ}{\Upsilon}
\newunicodechar{φ}{\varphi}
\newunicodechar{ϕ}{\phi}
\newunicodechar{Φ}{\Phi}
\newunicodechar{χ}{\chi}
\newunicodechar{ψ}{\psi}
\newunicodechar{Ψ}{\Psi}
\newunicodechar{ω}{\omega}
\newunicodechar{Ω}{\Omega}
\newunicodechar{𝟙}{\mathds{1}}
\newunicodechar{✗}{\ding{55}}

\usepackage{dsfont}
\newunicodechar{𝟙}{\mathds{1}}
\usepackage{bm}
\usepackage{xspace}
\usepackage{todonotes}
\usepackage{silence}
\usepackage{tikz}
\usetikzlibrary{calc,arrows.meta,decorations.pathreplacing}

\makeatletter
\g@addto@macro\@floatboxreset\centering
\makeatother

    \theoremstyle{plain}
    
    \crefname{conjecture}{Conjecture}{Conjectures}
    \newtheorem{fact}{Fact}
    \crefname{fact}{Fact}{Facts}

\newcommand{\whp}{\textrm{whp}\xspace}
\renewcommand{\O}{\mathcal{O}}

\DeclareMathOperator{\Po}{Po}
\DeclareMathOperator{\Bin}{Bin}

\DeclareMathOperator{\E}{\mathbb{E}}

\def\e{\mathrm{e}}
\def\U{\mathcal{U}}

\def\ms{\textrm{\textmu s}}
\def\ns{\textrm{ns}}

\def\row{\textsf{row}}
\def\piv{\mathrm{piv}}
\def\pos{\mathrm{pos}}

\def\refgauss{\texorpdfstring{{\small\upshape\sffamily SGAUSS}}{SGAUSS}\xspace}
\def\refcfrh{\texorpdfstring{{\small\upshape\sffamily CFRH}}{CFRH}\xspace}
\def\construct{\mbox{\upshape\sffamily construct}\xspace}
\def\query{\mbox{\upshape\sffamily query}\xspace}

\newcommand{\ignore}[1]{}

\title{Efficient Gauss Elimination for Near-Quadratic Matrices
with One Short Random Block per Row, with Applications}
\titlerunning{Near-Quadratic Matrices with One Short Random Block per Row}

\author{Martin Dietzfelbinger}{Technische Universität Ilmenau, Germany}{martin.dietzfelbinger@tu-ilmenau.de}{https://orcid.org/0000-0001-5484-3474}{}
\author{Stefan Walzer}{Technische Universität Ilmenau, Germany}{stefan.walzer@tu-ilmenau.de}{https://orcid.org/0000-0002-6477-0106}{}
\authorrunning{M.\,Dietzfelbinger and S.\,Walzer}
\Copyright{Martin Dietzfelbinger and Stefan Walzer}

\ccsdesc[500]{Theory of computation~Data structures design and analysis}
\keywords{Random Band Matrix, Gauss Elimination, Retrieval, Hashing, Succinct Data Structure, Randomised Data Structure, Robin Hood Hashing, Bloom Filter}
\category{}
\relatedversion{}
\supplement{}
\funding{}
\acknowledgements{}

\relatedversion{This paper will be presented at the European Symposium on Algorithms 2019.}
\hideLIPIcs
\nolinenumbers

\let\paragraph=\subparagraph
\begin{document}
\maketitle

\begin{abstract}
    	In this paper we identify a new class of sparse near-quadratic random Boolean matrices
			that have full row rank over $\mathbb{F}_2=\{0,1\}$ with high probability and can be transformed into echelon form
			in almost linear time by a simple version of Gauss elimination.
			The random matrix with dimensions $n(1-ε) \times n$
			is generated as follows: In each row, identify a block of length $L=O((\log n)/ε)$
			at a random position. The entries outside the block are 0, the entries inside the block are given by fair coin tosses. 
			Sorting the rows according to the positions of the blocks transforms the matrix into a kind of band matrix,
			on which, as it turns out, Gauss elimination works very efficiently with high probability. For the proof, the effects of Gauss elimination
			are interpreted as a (``coin-flipping'') variant of Robin Hood hashing, whose behaviour can be 
			captured in terms of a simple Markov model from queuing theory. Bounds for expected construction time
			and high success probability follow from results in this area. They readily extend to larger finite fields in place of $\mathbb{F}_2$.	
				
			By employing hashing, this matrix family leads to a new implementation of a \emph{retrieval} data structure, 
			which represents an arbitrary function $f\colon S \to \{0,1\}$ for some set $S$ of $m=(1-ε)n$ keys. 
			It requires $m/(1-ε)$ bits of space, construction takes $\O(m/ε^2$) expected time on a word RAM, while queries take $\O(1/ε)$ time
			and access only one contiguous segment of $\O((\log m)/ε)$ bits in the representation
			($\O(1/ε)$ consecutive words on a word RAM). The method is readily implemented and highly practical, and it is competitive 
			with state-of-the-art methods.  In a more theoretical variant, which works only for unrealistically large $S$, 
			we can even achieve construction time $\O(m/ε)$ and query time $\O(1)$,
			accessing  $\O(1)$ contiguous memory words for a query.
			By well-established methods the retrieval data structure leads to efficient constructions of 
			(static) perfect hash functions and (static) Bloom filters with almost optimal space and very local storage access patterns for queries. 				
\end{abstract}

\section{Introduction}\label{sec:intro}

\subsection{Sparse Random Matrices}\label{subsec:intro:systems}
In this paper we introduce and study a new class of sparse random matrices over 
finite fields, which give rise to linear systems that are efficiently solvable with high probability%
\footnote{Events occur ``with high probability (whp)'' if they occur with probability $1-\O(m^{-1})$.}%
. For concreteness and ease of notation, we describe the techniques for the field $\mathbb{F}_2=\{0,1\}$.
(The analysis applies to larger fields as well, as will be discussed below.)
A matrix $A$ from this class has $n$ columns and $m=(1-\varepsilon)n$ rows for some small $\varepsilon > 0$.
We always imagine that a right hand side $\vec b\in\{0,1\}^m$ is given and that we wish to solve the system $A\vec z=\vec b$
for the vector of unknowns $\vec z$. 

\def\R{\mathcal{R}}
The applications (see \cref{subsec:retrieval}) 
dictate that the rows of $A$ are stochastically independent and are all chosen according to the same distribution $\R$ on $\{0,1\}^n$.
 Often, but not always, $\R$ is the uniform distribution on some pool $R ⊆ \{0,1\}^n$ of admissible rows. The following choices were considered in the literature.

\begin{enumerate}[(1)]
        • If $R = \{0,1\}^n$, then $A$ has full row rank whp for any $ε = ω(1/n)$. 
				In fact, the probability for full row rank is $>0.28$ even for $ε = 0$, 
				see {e.g.}~\cite{Cooper2000,P:An_Optimal:2009}. Solving time is $\tilde{\O}(n³)$.
        • A popular choice for $R$ is the set of vectors with 1's in precisely $k$ positions, 
				for constant $k$. Then $ε =\mathrm{e}^{-θ(k)}$ is sufficient for solvability 
				whp~\cite{PS:The:Satisfiability:2016}. 
				Solving time is still $\tilde{\O}(n³)$ if Gauss elimination is used and $\O(n²)$ if Wiedemann's algorithm~\cite{W:Solving:1986} is used, 
				but heuristics exploiting the sparsity of $A$ help considerably \cite{Vigna:Fast-Scalable-Construction-of-Functions:2016}.
       • In the previous setting with $k=3$ and $\varepsilon 
        \ge 0.19$, linear running time can be achieved with a simple greedy algorithm,
        since then the matrix can be brought into echelon form by row and column \emph{exchanges} 
				alone~\cite{BPZ:Practical:2013,HMWC:Graphs:1993,Molloy05:Cores-in-random-hypergraphs}. 
        Using $k > 3$ is pointless here, as then the required value of $\varepsilon$ increases. 
				• Luby et al.~\cite{LMSS:Efficient_Erasure:2001,LMSSS:Loss-Resilient:1997} study ``loss-resilient codes'' based on certain random bipartite graphs.				
				Translating their considerations into our terminology shows that at the core of their construction 
				is a distribution on $(1-ε)n\times n$-matrices with randomly chosen sparse rows as well. 
				Simplifying a bit, a number $D=\O(1/\varepsilon)$ is chosen and a weight sequence is carefully selected
				that will give a row at most $D$ many 1's and on average $O(\log D)$ many 1's (in random positions). 
				It is shown in~\cite{LMSS:Efficient_Erasure:2001,LMSSS:Loss-Resilient:1997} that such matrices not only have full row rank with high probability, 
				but that, as in \textsf{(3)}, row and column \emph{exchanges} suffice to obtain an echelon form whp.
				This leads to a solving time of $\O(n\log(1/ε))$ for the corresponding linear system. 		
				\item The authors of the present work describe in a simultaneous paper~\cite{DieWal:peeling:2019} the construction of 
				sparse $(1-\varepsilon)n\times n$ matrices for very small (constant) $\varepsilon$, with a fixed number of 1's per row, 
				which also allow solving the corresponding system by row and column exchanges. 
				(While behaviour in experiments is promising, determining the behaviour of the construction for arbitrarily small $\varepsilon$ is an open problem.) 
			• In a recent proposal \cite{DW:Retrieval-log-extra-bits:2019} by the authors of the present paper, 
			  a row $r \sim \R$ contains two \emph{blocks} of $\Theta(\log n)$ random bits
        at random positions (block-aligned) in a vector otherwise filled with 0's.
        It turned out that in this case even $\varepsilon = \O((\log n) / n)$ will give solvability with high probability. 
        Solution time is again about cubic (Gauss) resp. quadratic (Wiedemann), with heuristic pre-solvers cushioning the blow partly in practice.
				\end{enumerate}
Motivated by the last construction, we propose an even simpler choice for the distribution $\R$: A row $r$ consists of 0's except for 
one randomly placed block of some length $L$, which consists of random bits. 
It turns out that $L=\O((\log n)/\varepsilon)$ is sufficient to achieve solvability with high probability. 
The $L$-bit block fits into $\O(1)$ memory words as long as $\varepsilon$ is constant. 
Our main technical result (Theorem~\ref{mainProp}) is that the resulting random matrix has full row rank whp. 
Moreover, if this is the case then sorting the rows by starting points of the blocks 
followed by a simple version of Gauss elimination produces an echelon form of the matrix and a solution to the linear system.
The expected number of field operations is $\O(n L /\varepsilon)$, 
which translates into expected running time $\O(n/\varepsilon^2)$ on a word RAM.
For the proof, we establish a connection to a particular version of Robin Hood hashing, 
whose behaviour in turn can be understood by reducing it to a well-known situation in queuing theory. 
(A detailed sketch of the argument is provided in~\cref{subsec:proof:ideas}.)

To our knowledge, this class of random matrices has not been considered before. 
However, \emph{deterministic} versions of matrices similar to these random ones
have been thoroughly studied in the last century, for both infinite and finite fields. 
Namely, sorting the rows of our matrices yields matrices that
with high probability resemble \emph{band matrices},
where the nonzero entries in row $i$ are within a restricted range around column $\lfloor i/(1-\varepsilon)\rfloor$.
In the study of band matrices one usually has $\varepsilon=0$ and assumes that the matrix is nonsingular. 
Seemingly the best known general upper time bound for 
the number of field operations needed for solving band quadratic systems with bandwidth $L$ are $\O(n L^{\omega-1})=\O(n ((\log n)/\varepsilon)^{\omega-1})$,
where $\omega$ is the matrix multiplication exponent, see~\cite{Eberly:1992:BandMatrix,GolVLoan:1996:MatrixComp,PanSobzeAtinp:2001:BandedMatrices}.

\subsection{Retrieval}\label{subsec:retrieval}
One motivation for studying random systems as described above 
comes from data structures for solving the \emph{retrieval problem}, which can be described as follows:
Some ``\emph{universe}'' $\U$ of possible keys is given, 
as is a function $f \colon S\to W$, where $S \subseteq \U$ has finite size $m$ and $W=\{0,1\}^r$ for some $r\ge1$. 
A \emph{retrieval data structure} \cite{BPZ:Simple:2007,CKRT:The_Bloomier:2004,DP:Succinct:2008,P:An_Optimal:2009} makes it possible to recover $f(x)$ quickly for arbitrary given $x\in S$.
We do not care what the result is when $x \notin S$, which makes the retrieval situation different
from a dictionary, where the question ``$x\in S\,$?'' must also be decided.
A retrieval data structure consists of 
\begin{itemize}
	\item an algorithm $\construct$, which takes $f$ as a list of pairs
  (and maybe some parameters) as input and constructs an object $\textsf{DS}_f$, and
\item an algorithm $\query$, which on input $x\in \U$ and $\textsf{DS}_f$ outputs an element of $W$,
with the requirement that $\query(\textsf{DS}_f,x)=f(x)$ for all $x\in S$. 
\end{itemize}
The essential performance parameters of a retrieval data structure are:
\begin{itemize}
	\item the space taken up by $\textsf{DS}_f$ (ideally $(1+ε)m$ bits of memory for some small $ε>0$),
	\item the running time of $\construct$ (ideally $\O(m)$), and
	\item  the running time of $\query$ (ideally a small constant in the worst case, possibly dependent on $\varepsilon$, 
	and good cache behaviour).
\end{itemize}		
		In this paper we 
		concentrate on the case most relevant in practice, 
		namely the case of small constant $r$, in particular on%
		\footnote{Every solution for this case gives a solution for larger $r$ as well, with a slowdown not larger than $r$.
			In our case, this slowdown essentially only affects queries, not construction, 
			since the Gauss elimination based algorithm can trivially be extended 
			to simultaneously handle $r$ right hand sides $\smash{\vec b₁,…,\vec b_r}$ and produce $r$ solution vectors $\smash{\vec z₁,…,\vec z_r}$.
			This change slows down $\construct$ by a factor of $1+r/L = 1+\O(\varepsilon r/\log n)$.}
 $r=1$.

A standard approach is as follows \cite{BPZ:Simple:2007,CKRT:The_Bloomier:2004,DP:Succinct:2008,P:An_Optimal:2009}. 
Let $f\colon S \to \{0,1\}$ be given and let $n = m/(1-\varepsilon)$ for some $\varepsilon>0$. Use hashing
to construct a mapping
$\row\colon \U \to \{0,1\}^n$ such that $(\row(x))_{x ∈ S}$ is (or behaves like) a family of independent random variables drawn from a suitable distribution $\R$ on $\{0,1\}^n$.
Consider the linear system $(\langle \row(x),\vec z \,\rangle=f(x))_{x\in S}$. 
In case the vectors $\row(x)$, $x\in S$, are linearly independent, this system is solvable for $\vec z$.
Solve the system and store the bit vector $\vec z$ of $n$ bits (and the hash function used) as $\textsf{DS}_f$.
Evaluation is by $\query(\textsf{DS}_f,x)=\langle \row(x),\vec z \, \rangle$, for $x\in\U$.
The running time of \textsf{construct} is essentially the time for solving the linear system,
and the running time for \textsf{query} is the time for evaluating the inner product. 

A common and well-explored trick for reducing the construction time 
\cite{BKZ:A_Practical:2005,DR:Applications:2009,P:An_Optimal:2009,Vigna:Fast-Scalable-Construction-of-Functions:2016} is to split the key set into ``chunks'' of 
size $\Theta(C)$ for some suitable $C$ and constructing separate retrieval structures for the 
chunks. The price for this is twofold: In queries, one more hash function must be evaluated
and the proper part of the data structure has to be located;
regarding storage space one needs an array of $\Omega(m/C)$ pointers. 
In this paper, we first concentrate on a ``pure'' construction.
The theoretical improvements possible by applying the splitting technique will be discussed briefly in~\cref{sec:input:partitioning}.
The splitting technique is also used in experiments for our construction
in \cref{sec:experiments} to keep the block length small. 
In this context it will also be noted the related ``split-and-share'' technique from~\cite{DW07:Balanced:2007,DR:Applications:2009}
can be used to get rid of the assumption that fully random hash functions are available for free.

Our main result regarding the retrieval problem follows effortlessly from the analysis of the new random linear systems (formally stated as \cref{mainProp}). 
\begin{theorem}
    \label{thm:main}
    Let $\U$ be a universe. Assume the context of a word RAM with oracle access to fully random hash functions on $\U$. 
		Then for any $ε > 0$ there is a retrieval data structure such that for all $S ⊆ \U$ of size $m$
    \begin{enumerate}[{\upshape(i)}]\setlength\itemsep{0em}
            • \label{item:construct-succeeds} \construct succeeds with high probability.
            • \construct has expected running time $\O(\frac{m}{ε²})$.
            • The resulting data structure $\mathsf{DS}_f$ occupies at most $(1+ε)m$ bits.
            • \query has running time $\O(\frac{1}{ε})$ and accesses $\O(\frac{1}{ε})$ consecutive words in memory.
    \end{enumerate}
\end{theorem}


\subsection{Machine Model and Notation}
For a positive integer $k$ we denote $\{1,\dots,k\}$ by $[k]$.
The number $m$ always denotes the size of a domain -- the number of keys to hash, 
the size of a function for retrieval or the number of rows of a matrix.
A (small) real number $ε>0$ is also given.
The number $n$ denotes the size of a range. We usually have $m=(1-\varepsilon)n$.
In asymptotic considerations we always assume that $ε$ is constant and $m$ and $n$ tend to $\infty$,
so that for example the expression $\O(n/\varepsilon)$ denotes  a function that is bounded by $cm/\varepsilon$ for a constant $c$, 
for all $m$ bigger than some $m(\varepsilon)$.
By $\langle \vec y,\vec z\,\rangle$ we denote the inner product of two vectors $ \vec y$ and $\vec z$.
As our computational model we adopt the word RAM with memory words comprising $\Omega(\log m)$ bits,
in which an operation on a word takes constant time. In addition to AC$_0$ instructions
we will need the \textsc{parity} of a word as an elementary operation. For simplicity we assume this can be carried out in constant time,
which certainly is realistic for standard word lengths like 64 or 128.
In any case, as the word lengths used are never larger than $\O(\log m)$,
one could tabulate the values of \textsc{parity} for inputs of size $\frac12\log m$ 
in a table of size $\O(\sqrt{m}\log m)$ bits to achieve constant evaluation time for inputs comprising a constant number of words.

\subsection{Techniques Used}\label{subsec:techniques}
We use \emph{coupling} of random variables $X$ and $Y$ (or of processes $(X_i)_{i\ge1}$ and $(Y_i)_{i\ge1}$).
By this we mean that we exhibit a single probability space on which $X$ and $Y$ (or $(X_i)_{i\ge1}$ and $(Y_i)_{i\ge1}$) are defined,
so that there are interesting \emph{pointwise} relations between them, like $X\le Y$, or $X_i \le Y_i + a$ for all $i\ge1$, for a constant $a$.
Sometimes these relations hold only conditioned on some (large) part of the probability space. We will make use of the following observation.
If we have random variables $U_0,\dots,U_k$ with couplings, {i.e.} joint distributions, of $U_{\ell-1}$ and $U_\ell$, for $1\le \ell \le k$,
then there is a common probability space on which all these random variables are defined
and the joint distribution of $U_{\ell-1}$ and $U_\ell$ is as given.%
\footnote{We do not prove this formally, since arguments like this belong to basic probability theory or measure theory.
The principle used is that the pairwise couplings give rise to conditional expectations $\E(U_\ell \mid U_{\ell-1})$.
Arguing inductively, given a common probability space for $U_1,\dots,U_{\ell-1}$ and $\E(U_\ell \mid U_{\ell-1})$, 
one can obtain a common probability space for $U_1,\dots,U_\ell$ so that 
$(U_1,\dots,U_{\ell-1})$ is distributed as before and $\E(U_\ell \mid U_1,\dots,U_{\ell-1}) = \E(U_\ell \mid U_{\ell-1})$.
-- This is practically the same as the standard argument that shows that a sequence of conditional expectations 
gives rise to a corresponding Markov chain on a joint probability space.} 

\section{Random Band Systems that Can be Solved Quickly}\label{sec:random:systems}

The main topic of this paper are matrices generated by the following random process. 
Let $0 < ε < 1$ and $n ∈ ℕ$. For a number $m = (1-ε)n$ of rows and some number $L\ge1$ 
we consider a matrix $A=(a_{ij})_{i∈[m],\,j∈[n+L-1]}$ over the field $\mathbb{F}₂$, 
chosen at random as follows. For each row $i ∈ [m]$ a \emph{starting position} $s_i ∈ [n]=\{1,\dots,n\}$ is chosen uniformly at random. 
The entries $a_{ij}$, $s_i ≤ j < s_i+L$ form a \emph{block} of fully random bits, all other entries in row $i$ are $0$.

In this section we show that for proper choices of the parameters such a random matrix will have full row rank 
and the corresponding systems $A\vec z=\vec b$ will be solvable very efficiently whp. 
Before delving into the technical details, we sketch the main ideas of the proof. 

\subsection{Proof Sketch}\label{subsec:proof:ideas}
As a starting point, we formulate a simple algorithm,
a special version of Gaussian elimination,
for solving linear systems $A\vec z=\vec b$ as just described.
We first sort the rows of $A$ by the starting position of their block. 
The resulting matrix resembles a band matrix, and we apply standard Gaussian elimination to it,
treating the rows in order of their starting position. 
Conveniently, there is no ``proliferation of 1's'', {i.e.} we never produce a 1-entry outside of any row's original block. 
In the round for row $i$, the entries $a_{ij}$ for $j=s_i,\dots, s_i + L-1$ are scanned. 
If column $j$ has been previously chosen as pivot then $a_{ij} = 0$. 
Otherwise, $a_{ij}$ is a random bit. While this bit may depend in a complex way on
the original entries of rows $1,\dots,i$ (apart from position $(i,j)$), 
for the analysis we may simply imagine that $a_{ij}$ is only chosen now by flipping a fair coin. 
This means that we consider eligible columns from left to right,  and the first $j$ for which the coin flip turns up 1 becomes the pivot column for row $i$. 
This view makes it possible to regard choosing pivot columns for the rows 
as probabilistically equivalent to a slightly twisted version of Robin Hood hashing.
Here this means that $m$ keys $x_1,\dots,x_m$ with random hash values in $\{1,\dots,n+L-1\}$ are given and, in order of increasing hash values, are inserted 
in a linear probing fashion into a table with positions $1,\dots,n+L-1$ (meaning that for $x_i$ cells $s_i,s_{i+1},\dots$ are inspected). 
The twist is that whenever a key probes an empty table cell
flipping a fair coin decides whether it is placed in the cell or has to move on to the next one.
The resulting position of key $x_i$ is the same as the position of the pivot for row $i$.
As is standard in the precise analysis of linear probing hashing
we switch perspective and look at the process from the point of view of cells $1,2,\dots,n,\dots,n+L-1$.
Associated with position (``time'') $j$ is the set of keys that probe cell $j$ (the ``queue''), and the quantity to 
study is the length of this queue. It turns out that the average queue length determines the overall cost of the row additions,
and that the probability for the maximum queue length to become too large is decisive for bounding
the success probability of the Gaussian elimination process.   
The first and routine step in the analysis of the queue length is to “Poissonise” arrivals such that the evolution of the queue length becomes a Markov chain. 
A second step is needed to deal with the somewhat annoying possibility that 
in a cell all keys that are eligible for this cell reject it because of their coin flips. 
We end up with a standard-queue (an ``M/D/1 queue'' in Kendall notation) and can use existing results 
from queuing theory to read off the bounds regarding the queue length needed to complete the analysis.

The following subsections give the details. 

\subsection{A Simple Gaussian Solver}\label{subsec:sgauss}
\SetKwProg{algo}{Algorithm}{:}{}
\SetKw{Continue}{continue}
\SetKw{Break}{break}
\SetKw{With}{with}

We now describe the algorithm to solve linear systems involving the random matrices described above. 
This is done by a variant or Gauss elimination, which will bring the matrix into echelon form (up to 
leaving out inessential column exchanges) and then apply back substitution.  

Given a random matrix $A=(a_{ij})_{i∈[m],\,j∈[n+L-1]}$ as defined above, with blocks of length $L$
starting at positions $s_i$, for $i∈[m]$, as well as some $\vec b ∈ \{0,1\}^{m}$, 
we wish to find a solution $\vec{z}$ to the system $A\vec{z} = \vec{b}$. Consider algorithm \refgauss (\cref{algo:gauss}). 
If $A$ has linearly independent rows, it will return a solution $\vec{z}$ and produce intermediate values $(\piv_i)_{i ∈ [m]}$.
(These will be important only in the analysis of the algorithm.) If the rows of $A$ are linearly dependent, the algorithm will fail. 

\begin{figure}[htpb]
    \newcommand{\tikzmark}[1]{\tikz[overlay,remember picture] \node (#1) {};}
    \begin{algorithm}[H]
      \algo{{\normalfont\refgauss}$(A = (a_{ij})_{i∈[m],\,j∈[n+L-1]}, (s_i)_{i ∈ [m]}, \vec{b} ∈ \{0,1\}^m)$}{
        sort the rows of the system $(A,\vec{b})$ by $s_i$ (in time $\O(m)$)\;
        relabel such that $s₁ \le s₂ \le \dots \le s_m$\;
        $\piv₁, \piv₂, …,\piv_m ← 0$\;
        \For{$i = 1, …, m$}{
            \For{$j = s_i,…,s_i+L-1$\tikzmark{top}}{
                \If{$a_{ij} = 1$\tikzmark{bot}}{
                    $\piv_i ← j$\;
                    \For{$i'$ \With $i' > i ∧ s_{i'} ≤ \piv_i$}{
                        \If{$a_{i', \piv_i} = 1$}{
                            $a_{i'} ← a_{i'} ⊕ a_i$\ \  \tcp{row addition (= subtraction)}
                            $b_{i'} ← b_{i'} ⊕ b_i$\;
                        }
                    }
                    \Break
                }
            }
            \lIf{$\piv_i = 0$\ \ \tcp{row $i$ is 0}}{\Return \textsc{Failure}}
        }
        \tcp{back substitution:}
        $\vec{z} ← \vec{0}$\;
        \For{$i = m,…,1$}{
            $z_{\piv_i} ← \langle \vec{z}, a_i \rangle ⊕ b_i$\ \ \ \tcp{note: $a_{ij} = 0$ for $j$ outside of $\{s_i, …, s_i+L-1\}$}
        }
        \Return $\vec{z}$ \tcp{solution to $A\vec{z} = \vec{b}$}
      }%
    \caption{A simple Gaussian solver.}
    \label{algo:gauss}
    \end{algorithm}
    \begin{tikzpicture}[overlay, remember picture]
        \coordinate (topr) at ($(top)+(1cm,5pt)$);
        \coordinate (botr) at (topr |- bot);
        \draw [decoration={brace,amplitude=0.5em},decorate,thick,gray]
            (topr) --  (botr) node [align=left,pos=0.5, anchor=west] {\CommentSty{\ \ \ // \ search for leftmost $1$ in row $i$. Can be done}\\\CommentSty{\ \ \  // \ in time $\O(L/\log m)$ on a word RAM. }};
    \end{tikzpicture}
\end{figure}

Algorithm \refgauss starts by sorting the rows of the system $(A,\vec{b})$ by their starting positions $s_i$ in linear time, 
{e.g.} using counting sort~\cite[Chapter 8.2]{CLRS:Introduction-3rd:2009}. We suppress the resulting permutation in the notation, assuming $s₁ \le s₂ \le \dots \le s_m$. 
Rows are then processed sequentially. When row $i$ is treated, its leftmost $1$-entry is found, if possible, and the corresponding column index is called 
the \emph{pivot} $\piv_i$ of row $i$. Row additions are used to eliminate $1$-entries from column $\piv_i$ in subsequent rows. 
Note that this operation never produces nonzero entries outside of any row's original block, {i.e.}~for no row $i$ are there ever any $1$'s outside of the positions 
$\{s_i,…,s_i + L-1\}$. To see this, we argue inductively on the number of additions performed. Assume $i>1$ and row $i'$ with $i' < i$ is added to row $i$. 
By choice of $\piv_{i'}$ and the induction hypothesis, nonzero entries of row $i'$ can reside only in positions $\piv_{i'},…,s_{i'}+L-1$. 
Again by induction and since row $i$ contains a~$1$ in position $\piv_{i'}$, we have $s_{i} ≤ \piv_{i'}$; moreover we have $s_{i'} + L - 1 \le s_i + L - 1$, due to sorting. 
Thus, row $i'$ contains no $1$'s outside of the block of row $i$ and the row addition maintains the invariant.

If an all-zero row is encountered, the algorithm fails (and returns \textsc{Failure}). This happens if and only if the rows of $A$ are linearly dependent%
\footnote{Depending on $\vec{b}$, the system $A\vec{z} = \vec{b}$ may still be solvable. We will not pursue this.}. 
Otherwise we say that the algorithm succeeds. 
In this case a solution $\vec{z}$ to $A\vec{z} = \vec{b}$ is obtained by back-substitution. 

It is not hard to see that the expected running time of \refgauss is dominated by the expected cost of row additions.

The proof of the following statement, presented in the rest of this section, is the main technical contribution of this paper.

\begin{theorem}
    \label{mainProp}
    There is some $L = \O((\log m)/ε)$ such that a run of \refgauss on the random matrix $A=(a_{ij})_{i∈[m],\,j∈[n+L-1]}$ 
        and an arbitrary right hand side $\vec b ∈ \{0,1\}^{m}$ succeeds whp. The expected number of row additions
        is $\O(m/ε)$. Each row addition involves entries inside one block and takes time $\O(1/\varepsilon)$ on a word RAM. 
\end{theorem}


\subsection{Coin-Flipping Robin Hood Hashing}\label{subsec:robin:hood}
Let $\{x_1,\dots,x_m\} ⊆ \U$ be some set of \emph{keys} to be stored in a hash table $T$.
Each key $x_i$ has a uniformly random \emph{hash value} $h_i ∈ [n]$. 
An (injective) placement of the keys in $T$ fulfils the \emph{linear probing} requirement 
if each $x_i$ is stored in a cell $T[\pos_i]$ with $\pos_i ≥ h_i$ and all cells $T[j]$ for $h_i ≤ j < \pos_i$ are non-empty. 
In \emph{Robin Hood hashing} there is the additional requirement that $h_i > h_{i'}$ implies $\pos_i > \pos_{i'}$. 
Robin Hood hashing is interesting because it minimises the variance of the displacements $\pos_i-h_i$. It has been studied in detail in several papers%
~\cite{CelLarMun:1985:RobinHood,DevMorVio04,Janson2005,JV:Unified-Linear-Probing:2016,Viola:2005:ExactLinearProbing}. 

Given the hash values $(h_i)_{i\in[m]}$, a placement of the keys obeying the Robin Hood linear probing conditions can be obtained as 
follows: Insert the keys in the order of increasing hash values, by the usual linear probing insertion procedure, which
probes ({i.e.} inspects) cells $T[h_i], T[h_i+1], \dots$ until the first empty cell is found, and places $x_i$ in this cell. 
We consider a slightly “broken” variation of this method, which sometimes delays placements. In the placing procedure for $x_i$,
when an empty cell $T[j]$ is encountered, it is decided by flipping a fair coin
whether to place $x_i$ in cell $T[j]$ or move on to the next cell.
(No problem is caused by the fact that the resulting placement violates the Robin Hood requirement and even the linear probing requirement,
since the hash table is only used as a tool in our analysis.)
For this insertion method we assume we have an (idealised) unbounded array $T[1,2,\dots]$.   
The position in which key $x_i$ is placed is called $\pos_i$. At the end the algorithm itself checks whether
any of the displacements $\pos_i-h_i$ is larger than $L$, 
in which case it reports \textsc{Failure}.%
\footnote{The reason we postpone checking for \textsc{Failure} until the very end of the execution 
is that it is technically convenient to have the values $(\pos_i)_{i\in[m]}$ even if failure occurs.} 
\cref{algo:modified-rh} gives a precise description of this algorithm, which we term~\refcfrh.

\begin{figure}[htpb]
    \begin{algorithm}[H]
      \algo{\normalfont\refcfrh($\{x₁,…,x_m\} ⊆ \U$)}{
        sort $x₁,…,x_m$ by hash value $h₁,\dots,h_m$\;
        relabel such that $h₁ \leq \dots \leq h_m$\;
        $T ← [⊥,⊥,…]$ \tcp{empty array, ``$⊥$'' means ``undefined''}
        $\pos₁, \dots, \pos_m ← 0$\;
        \For{$i = 1, …, m$}{
            \For{$j = h_i,h_i+1,…$}{
                \If{$T[j] = ⊥ ∧ \mathrm{coinFlip}() = 1 \textsc{ (``heads'')}$}{
                    $\pos_i ← j$\;
                    $T[j] ← x_i$\;
                    \Break
                }
            }
        }
        \lIf{$∃i ∈ [m]: \pos_i - h_i ≥ L$}{\Return \textsc{Failure}}
        \Return $T$\;
      }
    \caption{The \emph{Coin-Flipping Robin Hood hashing} algorithm. Without the “coinFlip() = 1” condition, it would compute a Robin Hood placement with maximum displacement $L$, if one exists.}
    \label{algo:modified-rh}
    \end{algorithm}
\end{figure}

\subsection{Connection between \refgauss and \refcfrh}\label{subsec:relation:gauss:cfrh}

We now establish a close connection between the behaviour of algorithms \refgauss and \refcfrh,
thus reducing the analysis of \refgauss to that of \refcfrh.
The algorithms have been formulated in such a way that some structural similarity is immediate.   
A run of \refgauss on a matrix with random starting positions $(s_i)_{i ∈ [m]}$ and random entries yields a sequence of pivots $(\piv_i)_{i ∈ [m]}$; a run of \refcfrh on a key set with random hash values $(h_i)_{i ∈[m]}$ performing random coin flips yields a sequence of positions $(\pos_i)_{i ∈ [m]}$.
We will see that the distributions of $(\piv_i)_{i ∈ [m]}$ and $(\pos_i)_{i ∈ [m]}$ are essentially the same and that moreover
two not so obvious parameters of the two random processes are closely connected. 
For this, we will show that outside the \textsc{Failure} events we can use the probability space underlying 
algorithm \refgauss to describe the behaviour of algorithm \refcfrh.   
This yields a coupling of the involved random processes. 


The first step is to identify $s_i = h_i$ for $i ∈ [m]$ (both sequences are assumed to be sorted and then renamed).
The connection between $\pos_i$ and $\piv_i$ is achieved by connecting the coin flips of \refcfrh to 
certain events in applying \refgauss to matrix $A$. We construct this correspondence by induction on $i$.
Assume rows $1,\dots,i-1$ have been treated, $x_1,\dots,x_{i-1}$ have been placed, and $\piv_{i'} = \pos_{i'}$ for all $1 ≤ i' < i$.

Now row $a_i$ (transformed by previous row additions) is treated. 
It contains a 0 in columns that were previously chosen as pivots, so possible candidates for $\piv_i$ are only indices from 
$J_i := \{s_i,…,s_i+L-1\} \setminus \{\piv₁,…,\piv_{i-1}\}$. For each $j ∈ J_i$, the initial value of $a_{ij}$ was a random bit.
The bits added to $a_{ij}$ in rounds $1,\dots,i-1$ are determined by the original entries of rows $1,\dots,i-1$ alone.
We order the entries of $J_i$ as  $j^{(1)} < j^{(2)} \dots < j^{(|J_i|)}$. 
Then, conditioned on all random choices in rows $1,\dots,i-1$ of $A$, 
the current values $a_{i,j^{(1)}},\dots,a_{i,j^{(k)}}$ still form a sequence of fully random bits.
We use these random bits to run round $i$ of \refcfrh, in which $x_i$ is placed. 
Since each cell can only hold one key, and by excluding runs where finally \textsc{failure} is declared, we may focus on the empty cells with indices in 
$\{h_i,…,h_i+L-1\} \setminus \{\pos₁,…,\pos_{i-1}\}  = \{s_1,\dots,s_i+L-1\} \setminus \{\piv₁,…,\piv_{i-1}\}=J_i$. 
We use (the current value) $a_{ij}$ as the value of the coin flip for cell $j$, for $j=j^{(1)}, j^{(2)}, \dots, j^{(|J_i|)}$.
The minimal $j$ in this sequence (if any) with $a_{ij}=1$ equals $\piv_i$ and $\pos_i$. If all these bits are 0, algorithm \refgauss will fail immediately,
and 
key $x_i$ will be placed in a cell $T[j]$ with $j \ge h_i+L$, so \refcfrh will 
eventually fail as well. 

Thus we have established that the random variables needed to run algorithm \refcfrh (outside of \textsc{Failure})
can be taken to belong to the probability space defined by $(s_i)_{i\in[m]}$ and the entries in the blocks of $A$ for algorithm \refgauss,
so that  (outside of \textsc{Failure}) the random variables $\pos_i$ and $\piv_i$ are the same. 
In the following lemma we state this connection as Claim {(i)}. 
In addition, we consider other random variables central for the analysis to follow. 
First, we define the \emph{height} of position $j ∈ [n +L -1]$ in the hash table as
\[ H_j := \#\{ i ∈ [m] \mid h_i ≤ j < \pos_i\}.\]
This is the number of keys probing table cell $j$ without being placed in it,
either because the cell is occupied or because it is rejected by the coin flip. 
Claim {(ii)} in the next lemma shows that $\sum_{j\in[n+L-1]}H_j$ essentially determines the running time of \refgauss,
so that we can focus on bounding $(H_j)_{j ∈ ℕ}$ from here on.
Further, with Claim {(iii)}, we get a handle on the question how large we have to choose $L$ in order to keep the failure probability small.

\begin{lemma}
    \label{lem:connection-gauss-to-lp}
    With the coupling just described, we get
    \begin{enumerate}[{\upshape(i)}]\setlength\itemsep{0em}
            • \label{item:success-identified} \refgauss succeeds iff \refcfrh succeeds. On success we have $\piv_i = \pos_i $ for all $i ∈ [m]$.
            • \label{item:number-of-additions} A successful run of \refgauss performs at most $\sum_{j ∈ [n+L-1]} H_j$ row additions.
            • \label{item:success-if-heights-small} Conditioned on the event $\max_{j ∈ [n]} H_j ≤ L - 2\log m$, the algorithms succeed \whp.
    \end{enumerate}
\end{lemma}

\begin{proof}
                    (ii) 
                \def\adds{\mathsf{Add}}%
        \def\disps{\mathsf{Displ}}%
        (Note that a similar statement with a different proof can be found in~\cite[Lemma 2.1]{Janson2008}.) 
                Consider the sets $\adds ≔ \{(i,i') ∈ [m]² \mid \text{\refgauss adds row $i$ to row $i'$}\}$ and $\disps ≔ \{(i,j) ∈ [m] × [n+L-1] \mid h_i ≤ j < \pos_i\}$. 
                Since $H_j$ simply counts the pairs $(i,j)\in\disps$ with $i\in[m]$, we have $|\disps| = \sum_{j ∈ [n+L-1]} H_j$. 
                To prove the claim we exhibit an injection from $\adds$ into $\disps$.
        
        Assume $(i,i') ∈ \adds$. If $\pos_i < \pos_{i'}$, we map $(i,i')$ to $(i',\pos_i)$. This is indeed an element of $\disps$,
                since $h_{i'} ≤ \piv_i = \pos_i < \pos_{i'}$ (if $\piv_i$ were smaller than $s_{i'}$, row $i$ would not be added to row $i'$).
      On the other hand, if $\pos_i > \pos_{i'}$, we map $(i,i')$ to $(i,\pos_{i'})$. This is in $\disps$ since $h_i =s_i ≤ s_i' ≤ \pos_{i'} < \pos_{i}$ 
                (recall that rows are sorted by starting position).
        
        The mapping is injective since from the image of $(i,i') ∈ \adds$ we can recover the set $\{i,i'\}$ with the help of 
                the injective mapping $i \mapsto \pos_i$, $i∈[m]$. 
                The fact that $i < i'$ fixes the ordering in the pair. 

                (iii) In \refcfrh, for an arbitrary $i ∈ [m]$ consider the state before key $x_i$ probes its first position $j := h_i$. 
                Any previous key $x_{i'}$ with $i' < i$ has a hash value $h_{i'} ≤ h_i$. Hence it either was inserted in a cell $j' < j$ or it has probed cell $j$. 
                Since at most $H_j$ keys have probed cell $j$, at most $H_j$ positions in $T[j,\dots,j+L-1]$ are occupied and at least $2\log m$ are free. The probability that $x_i$ is not placed in this region is therefore at most $2^{-2\log m} = m^{-2}$. By the union bound we obtain a failure probability of $\O(1/m)$.\qedhere
\end{proof}

\newcommand{\Em}{E_{\ge m}}
\newcommand{\Emax}{E_{\text{max}\,Z}}

\subsection{Bounding Heights in \refcfrh by a Markov Chain}\label{subsec:heights:to:markov}

\def\Geom{\mathrm{Geom}}

\cref{lem:connection-gauss-to-lp} tells us that we must analyse the heights in the hashing process \refcfrh. In this subsection,
we use ``Poissonisation'' of the hashing positions to majorise the heights in \refcfrh by a Markov chain, {i.e.} a process that is oblivous to the past, 
apart from the current height. Poissonisation is a common step in the analysis of linear probing hashing, see~{e.g.}~\cite{Viola:2005:ExactLinearProbing}.
Further, we wish to replace randomized placement by deterministic placement: Whenever a key is available for a position, 
one is put there (instead of flipping coins for all available keys). 
By this, the heights may decrease, but only by a bounded amount whp. The details of these steps are given in this subsection. 

In analysing \refcfrh (without regard for the event \textsc{Failure}), it is inconvenient that the starting positions $h_i$
are determined by random choices with subsequent sorting. Position $j$ is hit by a number of keys
given by a binomial distribution $\Bin(m,\frac1n)$ with expectation $\frac{m}{n}=1-ε$,
but there are dependencies. 
We approximate this situation by 
``Poissonisation''~\cite[Sect. 5.4]{MU:Probability:2005}.
Here this means that we assume that cell $j\in[n]$ is hit by $k_j$ keys, independently for $j=1,\dots,m$, where $k_j \sim \Po(1-ε')$
is Poisson distributed, for $ε'=ε/2$. 
Then the total number $m' = \sum_{j ∈ [n]}k_j$ of keys is distributed as $m' \sim \Po((1-ε')n)$.
Given $k_1,\dots,k_n$, we can imagine we have $m'$ keys with nondecreasing hash values $(h_i)_{i ∈ [m']}$, and 
we can apply algorithm \refcfrh to obtain key positions $(\pos'_i)_{i ∈ [m']}$ in $\{1,2,\dots\}$ and cell heights $(H'_j)_{j\ge1}$.

Conveniently, with Poissonisation, the heights $(H'_j)_{j ∈ [n]}$ turn out to form a Markov chain. 
This can be seen as follows. Recall that $H'_{j-1}$ is the number of keys probing cell $j-1$ without being placed there. 
Hence the number of keys probing cell $j$ is $H'_{j-1} + k_j$. 
One of these keys will be placed in cell $j$, unless $H'_{j-1}+k_j$ coin flips all yield 0, 
so if $g_j \sim \Geom(\frac 12)$ is a random variable with geometric distribution with parameter $\frac12$
(number of fair coin flips needed until the first 1 appears) and $b_j$ is the indicator function $\mathds{1}_{\{g_j \,>\, H'_{j-1}+k_j\}}$, we have $H'_j = H'_{j-1}+k_j-1+b_j$. 
(Note that the case $H'_{j-1}+k_j=0$ is treated correctly by this description. Conditioned on $H'_{j-1}+k_j$, the
value $b_j$ is a Bernoulli variable.)
The Markov property holds since $H'_j$ depends only on $H'_{j-1}$ and the two “fresh” random variables $k_j$ and $g_j$.

The following lemma allows us to shift our attention from $(H_j)$ to $(H'_j)$.

\begin{lemma}
    \label{lem:poissonise} Let $m=(1-ε)n$ and $m'\sim\Po((1-ε')m)$ for $ε' = ε/2$.
		There is a coupling between an ordinary run of \refcfrh (with $m$, $n$ and $H_j$) and a Poissonised run (with $m'$, $n$ and $H_j'$)
		such that conditioned on the high probability event $\Em = \{m' \ge m\}$ we have $H'_j ≥ H_j$ for all $j ∈ [n+L-1]$.
\end{lemma}

\begin{proof}
    Because $ε$ and $ε'  = ε/2$ are constants, the event $\Em$ has indeed high probability, as can be seen by well-known concentration bounds for the 
		Poisson distribution ({e.g.}~\cite[Th. 5.4]{MU:Probability:2005}). 
		For $m_0 \ge m$ fixed the distribution of the number of hits in the cells in $T[1,\dots,n]$ conditioned on $\{m'=m_0\}$
		is the same as what we get by throwing $m_0$ balls randomly into $n$ bins~\cite[Th. 5.6]{MU:Probability:2005}. Thus,  
		we may assume the Poissonised run has to deal with the $m$ keys of the ordinary run 
		plus $m'-m$ additional keys with random hash values in $[n]$. We apply algorithm \refcfrh to both inputs.
		After sorting, the new keys are inserted in some interleaved way with the ordinary keys.
		Now if one of the ordinary keys $x$ probes an empty cell $T[j]$, we use the same coin flip in both runs to decide whether to place it there;
		for the probing of the additional keys we use new, independent coin flips. 
		With this coupling it is clear that for all ordinary keys $x$ the displacement ``(position of $x$) $-$ (hash value of $x$)''
		in the Poissonized run is at least as big as in the ordinary run. As the additional keys can only increase heights, $H'_j ≥ H_j$ follows.
\end{proof}
As a further simplification, we eliminate the geometrically distributed variable $g_j$ and the derived variable $b_j$ in the Markov chain $(H'_j)_{j \ge 0}$. 
For this, let $(X_j)_{j \ge 0}$ be the Markov chain defined as
\begin{equation}
X₀ := 0 \quad \text{ and } \quad X_{j} := \max(0,X_{j-1} + d_j - 1) \quad \text{for } j ≥ 1, 
\label{eq:500}
\end{equation}
where $d_j \sim \Po(1-ε'/2)$ are independent random variables.

\begin{lemma}
    \label{lem:M/B/1-to-M/D/1}
    There is a coupling between $(X_j)_{j \ge 0}$ and $(H'_j)_{j \ge 0}$ such that $X_j + \log(4/ε') ≥ H'_j$ for all $j ∈ [n+L-1]$.
\end{lemma}

\begin{proof}
    Assume wlog that $\log(1/ε')$ is an integer. Let $b'_j \sim \Po(ε'/2)$ be a random variable 
		on the same probability space as $g_j$ such that $g_j > \log(4/ε')$ implies $b_j' ≥ 1$. This is possible because
    \[ \Pr[g_j > \log(4/ε')] = 2^{-\log(4/ε')} = ε'/4 ≤ 1-e^{-ε'/2} = \Pr[b_j' ≥ 1].\]
    We then define $d_j ≔ k_j + b'_j$ which gives $d_j \sim \Po(1-ε'/2)$.
		Proceeding by induction, and using~\eqref{eq:500}, we can define $(X_j)_{j \ge 0}$ and $(H_j')_{j\ge 0}$ on a common probability space.
		Then we check $X_j + \log(4/ε') ≥ H'_j$, also by induction:
		In the case $H'_{j-1} + k_j ≤ \log(4/ε')$ we simply get
    \[ X_j + \log(4/ε') ≥ \log(4/ε') ≥ H'_{j-1} + k_j ≥ H'_{j-1} + k_j + b_j - 1 = H'_j.\]
    Otherwise we can use the inequality $b_j = \mathds{1}_{\{g_j \,>\, H'_{j-1} + k_j\}} ≤ \mathds{1}_{\{g_j \,>\, \log(4/ε')\}} ≤ b_j'$ to obtain
    \begin{align*} X_j + \log(4/ε') &≥ X_{j-1} + d_j - 1 + \log(4/ε') \stackrel{\text{(Ind.Hyp.)}}{≥} H'_{j-1} + d_j - 1 \\  
		&= H'_{j-1} + k_j + b_j' - 1 ≥ H'_{j-1} + k_j + b_j - 1 = H'_j. \qedhere
		\end{align*}
\end{proof}

\subsection{Enter Queuing Theory}\label{subsec:queuing}

It turns out that, in essence, the behaviour of the Markov chain $(X_j)_{j \ge 0}$ has been studied in the literature under the name ``M/D/1 queue'', which is Kendall notation~\cite{Kendall:QueuingTheory:1953} for queues with “\textbf{M}arkovian arrivals, \textbf{D}eterministic service times and \textbf{1} server''.
We will exploit what is known about this simple queuing situation in order to finish our analysis. 

Formally, an M/D/1 queue is a Markov process $(Z_t)_{t ∈ ℝ_{≥0}}$ in continuous time and discrete space $ℕ₀=\{0,1,2,\dots\}$. 
The random variable $Z_t$ is usually interpreted as the number of \emph{customers} waiting in a FIFO queue at time $t ∈ ℝ_{≥0}$. 
Initially the queue is empty ($Z₀ = 0$). Customers arrive independently, {i.e.} arrivals are determined by a Poisson process with a rate we set to $ρ = 1-ε'/2$
(which implies that the number of customers arriving in any fixed time interval of length 1 is $\Po(\rho)$-distributed).
The \emph{server} requires one time unit to process a customer which means that if $t ∈ ℝ_{≥ 0}$ is the time of the first arrival, 
then customers will leave the queue at times $t+1,t+2,…$ until the queue is empty again.

Now consider the discretisation $(Z_j)_{j ∈ ℕ₀}$ of the M/D/1 queue. For $j ≥ 1$, the number $d_j$ of arrivals in between two observations $Z_{j-1}$ and $Z_j$ 
has distribution $d_j \sim \Po(ρ)$, and one customer was served in the meantime if and only if $Z_{j-1} > 0$. We can therefore write
\[ Z_j = \begin{cases}
    d_j & \text{ if $Z_{j-1} = 0$},\\
    Z_{j-1} + d_j - 1 & \text{ if $Z_{j-1} > 0$}.
\end{cases}\]
By reusing the variables $(d_j)_{j\ge1}$ that previously occurred in the definition of $(X_j)_{j\ge0}$, 
we already established a coupling between the processes $(X_j)_{j \ge 0}$ and $(Z_j)_{j \ge 0}$. A simple induction suffices to show
\begin{equation}
    X_j = \max(0,Z_j - 1).\label{eq:X-to-Z}\text{ for all }j\ge0.
\end{equation}
Intuitively, the server in the $X$-process is ahead by one customer because customers are processed at integer times “just in time for the observation”.

The following results are known in queuing theory:

\begin{fact}
    \label{fact:M/D/1-queues}
    \begin{enumerate}[{\upshape(i)}]\setlength\itemsep{0em}
        • \label{item:exp-Z} The average number of customers in the $Z$-queue at time $t ∈ ℝ_{≥0}$ is
        \[ \E[Z_t] ≤ \lim_{\tau → ∞} \E[Z_\tau] = ρ + \tfrac 12 \left( \frac{ρ²}{1-ρ}\right) = Θ(1/ε). \]
        (Precise values are known even for general arrival distributions, see \cite[Chapter 5.4]{Cooper:QueingTheory-2nd;1972}.)
        • \label{item:tail-Z}
        \cite[Prop 3.4]{EZB:M/D/1-tails:2006} We have the following tail bound for the event $\{Z_t > k\}$ for any $k ∈ ℕ$:
        \[ \Pr[Z_t > k] ≤ \lim_{\tau → ∞} \Pr[Z_\tau > k] = \e^{-k · \Theta(ε)}\text{, for all }t\ge0.\]
    \end{enumerate}
\end{fact}

\subsection{Putting the Pieces Together}\label{subsec:pieces:together}

\def\refrel#1#2#3{\stackrel{\text{#1 \ref{#2}}}{#3}}
We now have everything in place to prove \cref{mainProp} regarding solving our linear systems. 

\begin{proof}[Proof of \cref{mainProp}]
    By the observation made in \cref{subsec:techniques}, we may assume that the 
    random variables $(H_j)_{j ∈ [n+L-1]}$, $(H'_j)_{j ∈ [n+L-1]}$, $(X_j)_{j \ge 0}$ and $(Z_j)_{j \ge 0}$ and the three corresponding couplings
    are realized on one common probability space. 

		By \cref{fact:M/D/1-queues}(\ref{item:tail-Z}) it is possible to choose $L = Θ((\log m)/ε)$ while
		guaranteeing $\Pr[Z_j > L/2] = \O(m^{-2})$ for all $j\ge0$. 
            
    By the choice of $L$ and the union bound, the event $\Emax = \{ ∀ j ∈ [n+L-1]\colon Z_j ≤ L/2\}$ occurs whp. 
		Conditioned on $\Emax$ and the high probability event $\Em$ from \cref{lem:poissonise} we have
    \[H_j \refrel{Lem.}{lem:poissonise}{≤} H_j' \refrel{Lem.}{lem:M/B/1-to-M/D/1}{≤} X_j + \log(4/ε') 
		\stackrel{\text{Eq. \ref{eq:X-to-Z}}}{≤} Z_j + \log(4/ε') \stackrel{\Emax}{≤} L/2 + \log(4/ε') ≤ L - 2\log m. \]
		By using \cref{lem:connection-gauss-to-lp}(\ref{item:success-if-heights-small}) we conclude that \refgauss succeeds with high probability.
    
    Along similar lines we get, for each $j\in[n+L-1]$:
    \begin{align*}
        E[H_j] &\refrel{Lem.}{lem:poissonise}{≤} \E[H'_j \mid \Em] ≤ \tfrac 1{\Pr[\Em]}\E[H'_j] \refrel{Lem.}{lem:M/B/1-to-M/D/1}{≤} 
				                                                                                                                                                    \tfrac 1{\Pr[\Em]}\E[X_j+\log(4/ε')]\\
        &\stackrel{\text{Eq. \ref{eq:X-to-Z}}}{≤} \tfrac 1{\Pr[\Em]}\E[Z_j+\log(4/ε')] \stackrel{\text{\cref{fact:M/D/1-queues}(\ref{item:exp-Z})}}{≤} 
				                                                                                                                                                                             \tfrac 1{\Pr[\Em]}(\O(1/ε)+\log(4/ε')) 
				= \O(1/ε).
    \end{align*}
    By \cref{lem:connection-gauss-to-lp}(\ref{item:number-of-additions}) the expected number of row additions performed by a successful run 
		of \refgauss is therefore at most $\E[\sum_{j ∈ [n+L-1]} H_j] = \O(m/ε)$. 
		Since unsuccessful runs happen with probability $\O(1/m)$ and can perform at most $mL$ additions (each row can only be the target of $L$ row additions), 
		the overall expected number of additions is not skewed by unsuccessful runs, hence is also in $\O(m/ε)$. This finishes the proof of~\cref{mainProp}.
\end{proof}

\begin{remark*}
The analysis described in this section works in exactly the same way if instead of $\mathbb{F}_2$ a larger finite field $\mathbb{F}$ is used. 
A row in the random matrix is determined by a random starting position and a block of $L$ random elements from $\mathbb{F}$. 
A row operation in the Gaussian elimination now consists of a division, a multiplication of a block with a scalar
and a row addition.  The running time of the algorithm will increase at least by a factor of $\log(\lvert\mathbb{F}\rvert)$ (the bitlength of a field element),
and further increases depend on how well word parallelism in the word RAM can be utilized for operations like row additions 
and scalar product computations. 
(In~\cite{Vigna:Fast-Scalable-Construction-of-Functions:2016}, efficient methods are described for the field of three elements.)
The queue length will become a little smaller, but not significantly, since 
even the M/D/1 queue with arrivals with a Poisson$(1-\varepsilon)$ distribution will lead to average queue length $\Theta(1/\varepsilon)$.
\end{remark*}

\begin{remark*}
An interesting question was raised by a reviewer of the submission: Is anything gained
if we fix the first bit of each block to be 1? When checking our analysis for this case we see
that this 1-bit need not survive previous row operations. However, such a change does improve success probabilities in 
the Robin Hood insertion procedure. If a key $x_i$ finds cell $h_i$ empty, it
occupies this cell, without a coin being flipped. From the point of view of the queues, 
we see that now the derived variable $b_j$ in~\cref{subsec:heights:to:markov} is 1 if $k_j>0$ and geometrically distributed only if $k_j=0$. 
As in the preceding remark, this brings the process closer to the M/D/1 queue with arrivals with a Poisson$(1-\varepsilon)$ distribution
and deterministic service time 1, but the average queue length remains $\Theta(1/\varepsilon)$.
Still, it may be interesting to check by experiments if improvements result by this change. 
\end{remark*}

\section{A New Retrieval Data Structure}    \label{sec:main:retrieval}

With \cref{mainProp} in place we are ready to carry out the analysis of 
the retrieval data structure based on the new random matrices as described in~\cref{subsec:retrieval}. 
The proof of \cref{thm:main} is more or less straightforward.

\begin{proof}[Proof of \cref{thm:main}]
    Denote the $m$ elements of $S$ by $x₁,…,x_m$, let $n = \frac{1}{1-ε}m$, $L = Θ(\frac{\log m}{ε})$ the number from \cref{mainProp} 
		and $h \colon \U → [n] × \{0,1\}^L$ a fully random hash function. 
		For \construct, we associate the values $(s_i,p_i) := h(x_i)$ with each $x_i$ for $i ∈ [m]$ 
		and interpret them as a random band matrix $A = (a_{ij})_{i ∈[m],\, j ∈ [n+L-1]}$, 
		where for all $i ∈ [m]$ row $a_i$ contains the pattern $p_i$ starting at position $s_i$ and 0's everywhere else. 
		Moreover, let $\vec{b} ∈ \{0,1\}^m$ be the vector with entries $b_i = f(x_i)$ for $i ∈ [m]$. 
		We call \refgauss (\cref{algo:gauss}) with inputs $A$ and $\vec{b}$, obtaining (in case of success) a solution 
		$\vec{z} ∈ \{0,1\}^{n+L-1}$ of $A\vec{z} = \vec{b}$. The retrieval data structure is simply $\mathsf{DS}_f = \vec{z}$.
    
    By \cref{mainProp} \construct succeeds \whp\footnotemark{} (establishing (i)) 
		and performs $\O(m/ε)$ row additions. Since additions affect only $L = \O(\frac{\log m}{ε})$ 
		consecutive bits, and since a word RAM can deal with $\O(\log m)$ bits at once, a single row addition costs $\O(1/ε)$ time, 
		leading to total expected running time $\O(m/ε²)$ (which establishes (ii)).
    \footnotetext{If success with probability $1$ is desired, then in case the construction fails with hash function $h₀ = h$, we just restart the construction with different hash functions $h₁,h₂,…$. In this setup, $\mathsf{DS}_f$ must also contain the seed $s ∈ ℕ₀$ identifying the first hash function $h_s$ that led to success.}
    
    The data structure $\mathsf{DS}_f = \vec{z}$ occupies exactly $\frac{1}{1-ε}m + L-1 < (1+2ε)m$ bits. Replacing $ε$ with $ε/2$ yields the literal result~(iii).
    
    To evaluate $\query(\mathsf{DS}_f,y)$ for $y ∈ \U$, we compute $(s_y,p_y) = h(y)$ and the scalar product 
		$b_y = \langle \vec{z}\,[s_y…s_y{+}L{-}1],p_y\rangle ≔ \bigoplus_{j = 1}^L \vec{z}_{s_y{+}j{-}1} · p_{yj}$. 
		By construction, this yields $b_i = f(x_i)$ in the case that $y = x_i$. 
		To obtain (iv), observe that the scalar product of two binary sequences of length $L = \O(\log(n)/ε)$ 
		can be computed using $\O(1/ε)$ bit parallel \textsc{and} and \textsc{xor} operations, 
		as well as a single \textsc{parity} operation on $\O(\log m)$ bits, which can be assumed 
		to be available in constant time.%
\end{proof}
\begin{remark*}As the proof of \cref{mainProp} remains valid for arbitrary fixed finite fields in place of $\mathbb{F}_2$,
the same is true for \cref{thm:main}.
This is relevant for the compact representation of functions with small ranges like $[3]$, where binary encoding 
of single symbols implies extra space overhead. Such functions occur in 
data structures for perfect hash functions~\cite{BPZ:Practical:2013,Vigna:Fast-Scalable-Construction-of-Functions:2016}.
\end{remark*}

\section{Input Partitioning} 
\label{sec:input:partitioning}

We examine the effect of a simple trick to improve construction and query times of our retrieval data structure. 
We partition the input into \emph{chunks} using a ``first-level hash function'' and construct a separate retrieval data structure for each chunk.
Using this with chunk size $C = m^\varepsilon$ will reduce the time bounds for construction and query by a factor of $\varepsilon$.
The main reason for this is that we can use smaller block sizes $L$, which in turn makes row additions and inner products cheaper. 
Note that the idea is not new. Partitioning the input has previously been applied in the context of retrieval to reduce construction times, 
especially when ``raw'' construction times are superlinear \cite{DP:Succinct:2008,Vigna:Fast-Scalable-Construction-of-Functions:2016,P:An_Optimal:2009} 
or when performance in external memory settings is an issue \cite{BBOVV:Cache-Oblivious-Peeling:14,BPZ:Practical:2013}. 
Partitioning also allows us to get rid of the full randomness assumption, which is interesting from a theoretical point of 
view~\cite{BPZ:Practical:2013,DW07:Balanced:2007,DR:Applications:2009}.

\begin{remark*}
The reader should be aware that the choice $C = m^{ε}$, which is needed to obtain a speedup of $1/\varepsilon$, 
is unlikely to be a good choice in practice and that this improvement only works for unrealistically large $m$. 
Namely, we use that $\frac{\log m}{m^ε} \ll ε$ for sufficiently large $m$. While the left term is indeed $o(1)$ and the right a constant, 
even for moderate values of $ε$ like 0.05 implausibly large values of $m$ are needed to satisfy the weaker requirement $\frac{\log m}{m^ε} < ε$. 
In this sense, \cref{cor:main} taken literally is of purely theoretical value. 
Still, the general idea is sound and it can give improvements in practice when partitioning less aggressively, say with $C \approx \sqrt{m}$.
For example, the good running times reported in~\cref{sec:experiments} are only possible with this splitting approach.
\end{remark*}

\begin{theorem}
    \label{cor:main}
    The result of \cref{thm:main} can be strengthened in the following ways.
    \begin{enumerate}[{\upshape(i)}]\setlength\itemsep{0em}
            • The statements of \cref{thm:main} continue to hold without the assumption of fully random hash functions being available for free. 
            • The expected construction time is $\O(m/ε)$ (instead of $\O(m/ε²)$).
            • The expected query time is $\O(1)$ (instead of $\O(1/ε)$). Queries involve accessing a (small) auxiliary data structure, so technically not all required data is “consecutive in memory”.
    \end{enumerate}
\end{theorem}
\begin{proof} (Sketch.)
  Let $C = m^{ε}$ be the \emph{desired chunk size}.
	In~\cite[Section 4]{BPZ:Practical:2013} it is described in detail
	how a splitting function can be used to obtain chunks that have size within a constant factor of $C$ with high
	probability, and how fully random hash functions on each individual chunk can be provided by a
	randomized auxiliary structure $\mathcal{H}$ that takes only $o(m)$ space. 
	New functions can be generated by switching to new seeds. 
	(This construction is a variation of what is described in \cite{DR:Applications:2009,DW07:Balanced:2007}.)
	This fully suffices for our purposes. 
	We construct an individual retrieval data structure for each chunk with $L=\O(\frac{\log C}{ε}) = \O(\log m)$. 
	Such a construction succeeds in expected time $\O(C/ε)$ with probability $1-\O(1/C)$. 
	In case the construction fails for a chunk, it is repeated with a different seed. 
	At the end we save the concatenation of all $m/C$ retrieval data structures, 
	the data structure $\mathcal{H}$ and an auxiliary array. This array contains, for each chunk, 
	the offset of the corresponding retrieval data structure
	in the concatenation and the seed of the hash function used for the chunk. 
	It is easy to check that the size of all auxiliary data is asymptotically negligible.
    
   The total expected construction time is $\O((m/C)\cdot C/ε)=O(m/ε)$, and since $L=\O(\log m)$, a retrieval query can be evaluated in constant time.
\end{proof}
\begin{remark*} The construction from~\cite{LMSSS:Loss-Resilient:1997} described in item \textsf{(4)} in the list in~\cref{subsec:intro:systems}
can also be transformed in a retrieval data structure. (This does not seem to have been explored up to now.)
The expected running time for \construct is $\O(m\log(1/\varepsilon))$ (better than our $\O(m/\varepsilon)$),
the expected running time for \query is $\O(\log(1/\varepsilon))$, with $\O(\log(1/\varepsilon))$ random accesses in memory.
(Worst case is $\O(1/\varepsilon)$.)
In our preliminary experiments, see~\cref{sec:experiments}, for $m=10^7$, both construction and query times 
of our construction seem to be able to compete well with the construction following~\cite{LMSSS:Loss-Resilient:1997}.
\end{remark*}

\section{Experiments}
\label{sec:experiments}

We implemented our retrieval data structure following the approach explained in the proof of~\cref{cor:main}, 
except that we used \texttt{MurmurHash3} \cite{Appleby:MurmurHash3:2012} for all hash functions. 
This is a heuristic insofar as we depart from the full randomness assumption of \cref{thm:main}.
We report\footnotemark\ running times and space overheads in \cref{tab:results}, with the understanding that a retrieval data structure 
occupying $N$ bits of memory in total and accommodating $m$ keys has overhead $\frac{N}{m}-1$.
\footnotetext{%
    Experiments were performed on a desktop computer with an Intel${}^{\textrm{\textregistered}}$\,Core i7-2600 Processor @~3.40GHz. 
		Following~\cite{Vigna:Fast-Scalable-Construction-of-Functions:2016}, we used as data set $S$ the first $m = 10^7$ URLs from the \texttt{eu-2015-host} dataset gathered by \cite{BMSV:Crawls:2014} with ${≈}\,$80 bytes per key%
. As hash function we used \textsc{MurmurHash3\_x64\_128} \cite{Appleby:MurmurHash3:2012}. Reported query times are averages obtained by querying all elements of the data set once and include the evaluation of murmur, which takes about $25$\,ns on average. The reported figures are medians of 5 executions. 
}
Concerning the choice of parameters, $L = 64$ has practical advantages on a 64-bit machine and $C = 10^4$ seems to go well with it experimentally. As $ε ∈ \{7\%, 5\%, 3\%\}$ decreases, the measured construction time increases as would be expected. This is partly due to the higher number of row additions in successful constructions, but also due to an increased probability for a chunk's construction to fail, which prompts a restart for that chunk with a different seed. Note that, in our implementation, querying an element in a chunk with non-default seed also prompts an additional hash function evaluation.
\paragraph{Competing Implementations.} For comparison, we implemented the retrieval data structures from \cite{BPZ:Practical:2013,DW:Retrieval-log-extra-bits:2019,Vigna:Fast-Scalable-Construction-of-Functions:2016} and the one arising from the construction in \cite{LMSS:Efficient_Erasure:2001}.
(The number $D$ in~\cref{tab:results} is the maximum number of 1's in a row; the average is then $\Theta(\log D)$.)

In \cite{BPZ:Practical:2013}, the rows of the linear systems contain three 1's in independent and uniformly random positions. If the number of columns is $n = m/(1-ε)$ for $ε > 18.2\%$, the system can be solved in linear time by row and column \emph{exchanges} alone. Compared to that method, we achieve smaller overheads at similar running times.

The approaches from \cite{Vigna:Fast-Scalable-Construction-of-Functions:2016} and \cite{DW:Retrieval-log-extra-bits:2019} are different 
in that they construct linear systems that require cubic solving time with Gaussian elimination. 
This is counteracted by partitioning the input into chunks as well as by a heuristic \emph{LazyGauss}-phase of the solver that eliminates many variables before the Method of Four Russians \cite{Bard:Algebraic-Cryptanalysis:2009} is used on what remains. Construction times are higher than ours, but the tiny space overhead achieved in \cite{DW:Retrieval-log-extra-bits:2019} is beyond the reach of our approach.
The systems considered in \cite{Vigna:Fast-Scalable-Construction-of-Functions:2016}  resemble those in \cite{BPZ:Practical:2013}, except at higher densities. 
The systems studied in \cite{DW:Retrieval-log-extra-bits:2019} resemble our systems, except with \emph{two} blocks of random bits per row instead of one.

We remark that our approach is easier to implement than those of 
\cite{DW:Retrieval-log-extra-bits:2019,Vigna:Fast-Scalable-Construction-of-Functions:2016} but more difficult than that of \cite{BPZ:Practical:2013}.


\begin{table}
    \begin{tabular}{ccrrr}
        \toprule
        & Configuration & Overhead & \construct $[\ms/\textrm{key}]$& \query $[\ns]$\\
        \midrule
        \cite{BPZ:Practical:2013} & $ε = 19\%$ & $23.5\%$ & 0.32 & 59\\
        $\langle$\textsc{new}$\rangle$ &$ε = 7\%, L = 64, C = 10^4$&8.8\%&0.24&52\\
        $\langle$\textsc{new}$\rangle$ &$ε = 5\%, L = 64, C = 10^4$&6.5\%&0.27&54\\
        $\langle$\textsc{new}$\rangle$ &$ε = 3\%, L = 64, C = 10^4$&4.3\%&0.43&61\\
        \cite{LMSS:Efficient_Erasure:2001} & $c = 0.9, D = 12$ & 11.1\% & 0.79 & 94\\
        \cite{LMSS:Efficient_Erasure:2001} & $c = 0.99, D = 150$ & 1.1\% & 0.87 & 109\\
        \cite{Vigna:Fast-Scalable-Construction-of-Functions:2016} &$ε = 9\%, k = 3, C = 10^4$&10.2\%&1.30&58\\
        \cite{Vigna:Fast-Scalable-Construction-of-Functions:2016} &$ε = 3\%, k = 4, C= 10^4$&3.4\%&2.20&64\\
        \cite{DW:Retrieval-log-extra-bits:2019} &$ε = 0.05\%, ℓ = 16, C=10^4$&0.25\%&2.47&56\\
        \bottomrule
    \end{tabular}
    \caption{Space overhead and running times per key of some practical retrieval data structures}
    \label{tab:results}
\end{table}

\section{Conclusion}
\label{sec:conclusion}

This paper studies the principles of solving linear systems given by a particular kind of sparse random matrices,
with one short random block per row, in a random position. The proof changing the point of view from 
Gaussian elimination to Robin Hood hashing and then to queuing theory. It might be interesting
to find an direct, simpler proof for the main theorem. 
Preliminary experiments concerning an application with retrieval data structures are promising.
The most intriguing property is that evaluation of a retrieval query requires accessing only one (short) block in memory.
  
The potential of the construction in practice should be explored more fully and systematically, 
taking all relevant parameters like block size and chunk size into consideration.
Constructions of perfect hash functions like in~\cite{BPZ:Practical:2013,Vigna:Fast-Scalable-Construction-of-Functions:2016} 
or Bloom filters that combine perfect hashing with fingerprinting~\cite{BM:Survey:2003,DP:Succinct:2008} might profit from our construction.

\paragraph{Acknowledgements} 
    We are very grateful to Seth Pettie, who triggered this research by asking an insightful question regarding ``one block'' while discussing the 
    two-block solution from~\cite{DW:Retrieval-log-extra-bits:2019}. 
    (This discussion took place at the Dagstuhl Seminar 19051 ``Data Structures for the Cloud and External Memory Data''.) 
		Thanks are also due to the reviewers, whose comments helped to improve the presentation. 

\bibliography{bibliographie}
 
\end{document}